\documentclass[11pt,a4paper]{article}

\usepackage{amsfonts,amsmath,amssymb,amsthm}
\usepackage{latexsym,amscd}
\usepackage{amsbsy}
\usepackage{CJK}
\usepackage{indentfirst,mathrsfs}
\usepackage{amsfonts,amsmath,amssymb,amsthm}
\usepackage{latexsym,amscd}
\usepackage{amsbsy}
\usepackage{color}
\usepackage{multirow}
\usepackage{makecell}
\usepackage{float}

\newtheorem{thm}{Theorem}
\newtheorem{lem}{Lemma}
\newtheorem{cor}{Corollary}

\newcommand{\genlegendre}[4]{%
  \genfrac{(}{)}{}{#1}{#3}{#4}%
  \if\relax\detokenize{#2}\relax\else_{\!#2}\fi
}
\newcommand{\legendre}[3][]{\genlegendre{}{#1}{#2}{#3}}

\newcommand{\Tr}{{\mathrm{Tr}}}
\newcommand{\fqtwo}{{\mathbb F}_{q^2}}
\newcommand{\fq}{{\mathbb F}_{q}}

\newcommand{\fpn}{{\mathbb F}_{p^n}}
\newcommand{\ftwon}{{\mathbb F}_{2^n}}

\newcommand{\fp}{{\mathbb F}_{p}}

\newcommand{\ftwo}{{\mathbb F}_{2}}

\begin{document}

\title{The differential spectrum and boomerang spectrum of a class of locally-APN functions}
\author{Zhao Hu, Nian Li, Linjie Xu, Xiangyong Zeng and Xiaohu Tang
\thanks{Z. Hu, N. Li, X. Zeng and X. Tang are with the Hubei Key Laboratory of Applied Mathematics, Faculty of Mathematics and
Statistics, Hubei University, Wuhan, 430062, China. L. Xu is with the WuHan Marine Communication Research Institute, Wuhan, 430000, China. X. Tang is also with the Information Security and National Computing Grid Laboratory, Southwest Jiaotong University, Chengdu, 610031, China. Email: zhao.hu@aliyun.com, nian.li@hubu.edu.cn, linjiexu@126.com, xzeng@hubu.edu.cn, xhutang@swjtu.edu.cn}
}
\date{}
\maketitle
\begin{quote}
  {\small {\bf Abstract:}   In this paper, we study the boomerang spectrum of the power mapping $F(x)=x^{k(q-1)}$ over $\fqtwo$, where $q=p^m$, $p$ is a prime, $m$ is a positive integer and $\gcd(k,q+1)=1$. We first determine the differential spectrum of $F(x)$ and show that $F(x)$ is locally-APN. This extends a result of [IEEE Trans. Inf. Theory 57(12):8127-8137, 2011] from $(p,k)=(2,1)$ to general $(p,k)$. We then determine the boomerang spectrum of $F(x)$ by making use of its differential spectrum, which shows that the boomerang uniformity of $F(x)$ is 4 if $p=2$ and $m$ is odd and otherwise it is 2. Our results not only generalize the results in [Des. Codes Cryptogr. 89:2627-2636, 2021] and [Adv. Math. Commun., doi: 10.3934/amc.2022046] but also extend the example $x^{45}$ over ${\mathbb F}_{2^8}$ in [Des. Codes Cryptogr. 89:2627-2636, 2021] into an infinite class of power mappings with boomerang uniformity 2.}

  {\small {\bf Keywords:} Boomerang spectrum, Differential spectrum, Locally-APN function.}

  {\small {\bf Mathematics Subject Classification} }  12E20 $\cdot$ 11T06 $\cdot$ 94A60
\end{quote}

\section{Introduction}
Let $\fpn$ be the finite field with $p^n$ elements and $\fpn^*=\fpn \setminus \{0\}$, where $p$ is a prime and $n$ is a positive integer. Substitution boxes (S-boxes), which can be seen as cryptographic functions over finite fields, play a very crucial role in the design of secure cryptographic primitives, such as block ciphers. Differential attack, proposed by Biham and Shamir \cite{BSA}, is one of the most fundamental cryptanalytic tools to assess the security of block ciphers. Let $F(x)$ be a function from  $\fpn$ to itself. The differential uniformity of $F(x)$ is defined as
\[\delta(F)=\max\limits_{a \in \fpn^*}\max\limits_{b \in \fpn}  \delta_F(a,b)\]
where $\delta_F(a,b)=|\{x \in \fpn| ~\mathbb{D}_{a}F(x)=b\}|$ and $\mathbb{D}_aF(x)=F(x+a)-F(x)$.
Differential uniformity is an important concept in cryptography introduced by \cite{NK} as it quantifies the degree of security of the cipher with respect to the differential attack if $F$ is used as an S-box in the cipher. The lower the quantity of $\delta(F)$, the stronger the ability of the function $F$ resisting the differential attack. The function $F$ is called a perfect nonlinear (PN) function if $\delta(F)=1$, and an almost perfect nonlinear (APN) function if $\delta(F)=2$. Observe that $\delta(F)$ is always even when $p=2$ since $x+a$ is also the solution of $\mathbb{D}_{a}F(x)=b$ if $x$ is a solution of $\mathbb{D}_{a}F(x)=b$. Thus APN functions over $\ftwon$ offer maximal resistance to differential attacks.

The differential spectrum of a function $F(x)$ over $\fpn$ is defined as the multiset
$\{\,\delta_F(a,b)\,:\, a \in \fpn^*,\, b\in \fpn\,\}.$
When $F(x)$ is  a power mapping, i.e., $F(x)=x^d$ for {a positive integer} $d$,  one sees that $\delta_F(a,b)=\delta_F(1,{b/{a^d}})$  for all $a\in \fpn^*$ and $b\in \fpn$.
That is to say,  the differential spectrum of $F(x)$ is completely determined by the values of $\delta_F(1,b)$ as $b$ runs through $\fpn$.
Therefore, the differential spectrum of a power mapping $F(x)$ over $\fpn$ with $\delta(F)=\delta$ can be simply defined as
\[\mathbb{DS}_{F} = \{\omega_0, \omega_1, \ldots, \omega_{\delta}\},\]
where $\omega_i=|\left\{b\in \fpn\mid \delta_F(1,b)=i\right\}|$ for $0\leq i\leq \delta$.
It is well-known that the identities
\begin{eqnarray}\label{identity-differential}
\sum_{i=0}^{\delta}\omega_{i}=p^n\;\;{\rm and}\;\;\sum_{i=0}^{\delta}i\omega_{i}=p^n
\end{eqnarray}
hold which are useful in computing the differential spectrum. A power function $F(x)$ over $\fpn$ is said to be locally-APN if $\max \{\delta_F(1,b): b\in \fpn \backslash \fp\} =2$. The reader is referred to \cite{BCCP} for details when $p=2$.

The differential spectrum of a nonlinear function not only plays an important role in cryptography \cite{BCCP1,BCCP,YXLHML} but also has wide applications in sequences \cite{DHKM}, coding theory \cite{BCCP1,CPJP} and combinatorial design \cite{TDXM}. Therefore, it is an interesting topic to completely determine the differential spectrum of a nonlinear function. The infinite families of power mappings with known differential spectra are listed in Table \ref{differential-table}.
%
\begin{table}[!htb]\footnotesize
\caption{The power mapping $F(x)=x^d$ over $\fpn$ with known differential spectrum}  \label{differential-table}
\renewcommand\arraystretch{0.8}
\setlength\tabcolsep{2pt}
\centering
\begin{tabular}{|c|c|c|c|c|c|}
\hline No.& $p$ & $d$                & Condition                   &  $\delta(F)$    & Refs.                  \\ \hline\hline
  1  & $2$ & $2^t+1$            &    $\gcd(t,n)=s$            &     $2^s$       &\cite{BCCP1,EMSS}            \\ \hline
  2  & $2$ & $2^{2t}-2^t+1$     & $\gcd(t,n)=s$, $n/s$ odd    &    $2^s$        &\cite{BCCP1}            \\ \hline
  3  & $2$ & $2^n-2$            &$n\geq 2$                    & $2$ or $4$      &\cite{BCCP1,EMSS}            \\ \hline
  4  & $2$ &$2^{2k}+2^k+1$      &$n=4k$                       & 4               &\cite{BCCP1,XYHM}       \\ \hline
  5  & $2$ &$2^t-1$             &$t=3,n-2$                    & 6 or 8          &\cite{BCCP}             \\ \hline
  6  & $2$ &$2^t-1$             &$t=n/2, n/2+1$, $n$ even     & $2^{n/2}-2$ or $2^{n/2}$  &\cite{BCCP}   \\ \hline
  7  & $2$ &$2^t-1$             &$t=(n-1)/2,(n+3)/2$, $n$ odd & $6$ or $8$      &\cite{BPLC}             \\ \hline
  8  & $2$ &$2^{3k}+2^{2k}+2^{k}-1$ & $n=4k$                  & $2^{2k}$        &\cite{LWZT}             \\ \hline
  9  & $2$ &$2^m+2^{(m+1)/2}+1$ &$n=2m$, $m\geq5$ odd         & $8$             &\cite{XYYP}             \\ \hline
  10 & $2$ &$2^{m+1}+3$         &$n=2m$, $m\geq5$ odd         & $8$             &\cite{XYYP}             \\ \hline
  11 & $3$ &$2\cdot 3^{(n-1)/2}+1$  &$n$ odd                  &$4$              &\cite{DHKM}             \\ \hline
  12 & $3$ &$(3^n-1)/2+2$       &$n$ odd                      &$4$              &\cite{JLLQ}              \\\hline
  13 & $5$ &$(5^n-3)/2$         &any $n$                      &$4$ or $5$       & \cite{YLCH}             \\\hline
  14 & $5$ &$(5^n-1)/2+2$         &any $n$                      &$3$       & \cite{PLZX}             \\\hline
  15 & $p$ odd &$(p^k+1)/2$     &$e=\gcd(n,k)$                &$(p^e-1)/2$ or $p^e+1$  &\cite{CHNC}       \\\hline
  16 & $p$ odd &$(p^n+1)/(p^m+1)+(p^n-1)/2$&$p \equiv 3 \pmod{4}$, $n$ odd, $m|n$ &$(p^m+1)/2$ &\cite{CHNC}\\\hline
  17 & $p$ odd &$p^{2k}-p^k+1$  &$\gcd(n,k)=e$, $n/e$ odd     &$p^e+1$          &\cite{YZWWHW,LRFC}       \\\hline
  18 & $p$ odd &$p^m+2$         & $n=2m$, $p>3$               &$4$         &\cite{MXLH}       \\\hline
  19 & $p$ odd &$p^n-3$         & any $n$                     &$ 1\leq  \delta(F) \leq 5$         &\cite{YXLHML,XZLH}       \\\hline
  20 & any     &$k(p^m-1)$        & $n=2m$, $\gcd(k,p^m+1)=1$  &$p^m-2$ (locally-APN)   & This paper       \\\hline
\end{tabular}
\end{table}

Boomerang attack is an important cryptanalysis technique introduced by Wagner in \cite{WAG} against block ciphers involving S-boxes. It can be considered as an extension of the differential attack \cite{BSA}.
At Eurocrypt 2018, Cid et al. \cite{CHPS} proposed a new tool called Boomerang Connectivity Table (BCT) to measure the resistance of a block cipher against the boomerang attack. In the same year, Boura et al. \cite{BCCA} introduced the boomerang uniformity based on BCT to quantify the resistance of a function against the boomerang attack. In 2019, Li et al. \cite{LQSL} presented an equivalent definition of the boomerang uniformity which can be generalized to any finite field. Namely, for $a,b\in \fpn^*$, let $\beta_{F}(a,b)$ denote the number of solutions in $\fpn^2$ of the following system of equations
\begin{eqnarray} \label{boomerang-eq}
\left\{\begin{array}{ll}
F(x)-F(y)=b,  \\
F(x+a)-F(y+a)=b.
\end{array} \right.
\end{eqnarray}
Then the boomerang uniformity of a function $F$ over $\fpn$ is defined as
\[\beta(F)=\max_{a\in \fpn^*} \max_{b\in \fpn^*} \beta_{F}(a,b).\]
The boomerang spectrum of a function $F$ over $\fpn$ is defined as
$\{\,\beta_F(a,b)\,:\, a \in \fpn^*,\, b\in \fpn^{*}\,\}$.
Similar to the differential spectrum, for a power mapping $F(x)$, it suffices to determine the values of $\beta_F(1,b)$ as $b$ runs through $\fpn^*$ to compute the boomerang spectrum of $F(x)$.
Thus the boomerang spectrum of a power mapping $F(x)$ over $\fpn$ with $\beta(F)=\beta$ can be simply defined as
\[\mathbb{BS}_{F} = \{\nu_0, \nu_1, \ldots, \nu_{\beta}\},\]
where $\nu_i=|\{b\in \fpn^*\mid \beta_F(b)=i\}|$ for $0\leq i\leq \beta$.

A function $F(x)$ with smaller $\beta(F)$ provides stronger security against boomerang-style attacks. It was known in \cite{CHPS,MTXM} that  $\beta_{F}(a,b)\geq \delta_{F}(a,b)$ for $a,b\in \ftwon$ which indicates $\beta(F)\geq \max_{a \in \ftwon^*}\max_{b \in \ftwon^*}  \delta_F(a,b)$.  Note that $\delta_{F}(a,0)=0$ for $a\in \ftwon^*$ if $F$ permutes $\ftwon$. This implies $\beta(F)\geq \delta(F)$ for a permutation $F$  over $\ftwon$ \cite{CHPS}. Specially, it was proved in \cite{CHPS,HPSP} that $\beta(F)=2$ if $\delta(F)=2$ for a function $F$ over $\ftwon$. The converse is also true for a permutation $F$ over $\ftwon$ \cite{CHPS} and it is not necessarily true for a non-permutation $F$ over $\ftwon$ \cite{HPSP}. When $p$ is odd, it was shown in \cite{JLLQ} that $\beta(F)=0$ if $F$ is a PN function. The reader is referred to \cite{BCCA,CVIM,EMSS,HPSP2,KMCLJ,LHXZ,LXZX,LLHQ,LQSL,MTXM,MMMM} for some recent results in this direction. Table \ref{boomerang-table} gives the results on the infinite families of power mappings with known boomerang uniformity.
\begin{table}[!htb]\footnotesize
\caption{The power mapping $F(x)=x^d$ over $\fpn$ with known boomerang uniformity}  \label{boomerang-table}
\renewcommand\arraystretch{0.8}
\setlength\tabcolsep{2pt}
\centering
\begin{tabular}{|c|c|c|c|c|c|c|}
\hline No.& $p$ & $d$      & Condition                               & $\beta(F)$  &Is $\mathbb{BS}_{F}$ known?  & Refs.                \\ \hline\hline
  1  & $2$ & $2^t+1$       &$s=\gcd(t,n)$                            & $2^s$ or $2^s(2^s-1)$  & Yes &\cite{BCCA,EMSS,HPSP2} \\ \hline
  2  & $2$ & $2^{2k}+2^k+1$&$n=4k$, $k$ odd                          & $\leq 24$    &   No         &\cite{CVIM}           \\ \hline
  3  & $2$ & $2^m-1$         &$n=2m$                     & $2$ or $4$   &    No        &\cite{HPSP}   \\ \hline
  4  & $2$ & $7$           &$n=2m$, $\gcd(3,m)=1$, $m>6$             & $10$         &   No         &\cite{ZHLZ}   \\ \hline
  5  & $2$ & $2^{m+1}-1$  &$n=2m$, $m>1$                             & $2^m+2$      & Yes &\cite{YZZZ}   \\ \hline
  6  & $3$ & $(3^n-1)/2+2$ &$n$ odd                            & $3$          & Yes &\cite{JLLQ}    \\ \hline
  7  & $p$ odd & $p^m-1$   &$n=2m$,  $p^m\not \equiv 2 \pmod{3}$ & $2$          &No &\cite{YLSF}  \\ \hline
  8  & $p$ odd & $((p^m+3)(p^m-1))/2$ &$n=2m$,  see \cite{YLSF}  & $2$          &No &\cite{YLSF}  \\ \hline
  9  & $p$ odd & $(p^n-3)/2$ &  $p^n\equiv 3 \pmod{4}$  & $\leq 6$        &No &\cite{YZZZ1}  \\ \hline
  10  & any & $p^n-2$       &any $n$                                  & $2\le \beta(F)\le 6$ & Yes &\cite{BCCA,JLLQ,EMSS}  \\ \hline
  11  & any & $k(p^m-1)$     &$n=2m$, $\gcd(k,p^m+1)=1$         & $2$ or $4$  & Yes &This paper   \\ \hline
\end{tabular}
\end{table}

In this paper, we study the differential spectrum and boomerang spectrum of the power mapping $F(x)=x^{k(q-1)}$ over $\fqtwo$, where $q=p^m$, $p$ is a prime, $m$ is a positive integer and $\gcd(k,q+1)=1$. We first determine the differential spectrum of $F(x)$ completely for any prime power $q$ and show that $F(x)$ is locally-APN. This contributes a new class of power mappings with known differential spectrum and extends the result of \cite[Theorem 7]{BCCP} from $(p,k)=(2,1)$ to general $(p,k)$. Based on the differential spectrum of $F(x)$, we then compute the boomerang spectrum of $F(x)$, which shows that the boomerang uniformity of $F(x)$ is $2$ or $4$. It should be noted that our results generalize those of \cite{HPSP,YLSF}, in which the boomerang uniformity of $F(x)=x^{k(q-1)}$ over $\fqtwo$ was studied for the case $p=2$, $k=1$ \cite{HPSP}, the case $p$ odd, $k=1$ with $q\not \equiv 2 \pmod{3}$ \cite{YLSF} and the case $p$ odd, $k=(q+3)/2$ with $q\not \equiv 2 \pmod{3}$ and $q\equiv 3 \pmod{4}$ \cite{YLSF}.
It can be seen from Table \ref{boomerang-table} that three classes of power mappings with known boomerang uniformity are covered by our results.
Moreover, our results explain the example $x^{45}$ over ${\mathbb F}_{2^8}$ (a locally-APN function with boomerang uniformity 2) found by Hasan, Pal and St\u{a}nic\u{a} \cite{HPSP} and generalize it into an infinite class of power functions with boomerang uniformity 2.


\section{Preliminaries}

From now on, we always assume that $n=2m$ and $q=p^m$, where $p$ is a prime and $m$ and $n$ are positive integers. Denote the unite circle of $\fqtwo$ by $U_{q+1}=\{x \in \fqtwo: x^{q+1} = 1\}$ throughout this paper.


When $p$ is an odd prime, for a given element $\alpha\in \fq^*$, $\alpha$ is a square of $\fq^*$ if $\alpha^{\frac{q-1}{2}}=1$ and $\alpha$ is a nonsquare of $\fq^*$ if $\alpha^{\frac{q-1}{2}}=-1$.


\begin{lem} \label{Issquare(-3)}
Let $p>3$ be an odd prime, $m$ be a positive integer and $q=p^m$. Then $-3$ is a nonsquare of $\fq^*$ if $q\equiv 2 \pmod{3}$ and a square of $\fq^*$ otherwise.
\end{lem}

\begin{proof}
A direct computation gives
\begin{eqnarray*}
(-3)^{\frac{p^m-1}{2}}=(-3)^{\frac{p-1}{2}(p^{m-1}+\cdots+p+1)}=\left\{
\begin{array}{ll}
1, &\mbox{if $m$ is even}; \\
(-3)^{\frac{p-1}{2}}, &\mbox{if $m$ is odd}.
\end{array} \right.
\end{eqnarray*}
Thus it suffices to consider whether $-3$ is a square of $\fp^*$ or not.
Note that $(-3)^{\frac{p-1}{2}}=(-1)^{\frac{p-1}{2}}\legendre{3}{p}$, where $\legendre{3}{p}$
denotes the Legendre symbol.
Using the Law of Quadratic Reciprocity \cite[Theorem 5.17]{Lidl} gives
\begin{eqnarray*}
(-3)^{\frac{p-1}{2}}=\legendre{p}{3}=\left\{
\begin{array}{ll}
1, &\mbox{if $p\equiv 1 \pmod{3}$}; \\
-1, &\mbox{if $p\equiv 2 \pmod{3}$}.
\end{array} \right.
\end{eqnarray*}
This completes the proof.
\end{proof}

As a direct result of \cite[Lemma 4]{TZLH}, the following lemma is useful for the subsequent section.

\begin{lem} \label{square-eq-p=2}
Let $n=2m$ be an even positive integer, $q=2^m$ and $b\in \fqtwo^{*}$. Then the equation $x^2+bx+\frac{1}{b^{q-1}}=0$ has two distinct solutions in $U_{q+1}$ if and only if $\Tr_{1}^{m}(\frac{1}{b^{q+1}})=1$, where $\Tr_{1}^{m}(\cdot)$ is the trace function from $\mathbb{F}_{2^m}$ to $\mathbb{F}_{2}$.
\end{lem}

The following result can be proved by Lemma 4 and Remark 1 of \cite{TZ2019}. For the reader's convenience, we provide a proof here.

\begin{lem} \label{square-eq-p>2}
Let $n=2m$ be an even positive integer, $q=p^m$, $p$ be an odd prime and $b\in \fqtwo^{*}$. Then $x^2+bx+\frac{1}{b^{q-1}}=0$ has two distinct solutions in $U_{q+1}$ if and only if $\theta$ is a nonsquare in $\fq^*$, where $\theta=\frac{b^{q+1}-4}{b^{q+1}}$.
\end{lem}

\begin{proof}
Let $\Delta=\frac{b^{q+1}-4}{b^{q-1}}$ be the discriminant of the quadratic equation $x^2+bx+\frac{1}{b^{q-1}}=0$. Observe that $b^{q+1}-4\in\fq$ and $\Delta$ is a square of $\fqtwo^*$ if $b^{q+1}-4\ne 0$. That is to say, $x^2+bx+\frac{1}{b^{q-1}}=0$ always has two distinct solutions $x_{1}=\frac{-b+\sqrt{\Delta}}{2}$ and $x_{2}=\frac{-b-\sqrt{\Delta}}{2}$
in $\fqtwo^*$ if and only if $b^{q+1}-4\ne 0$. Next we always assume that $b^{q+1}-4\ne 0$. Notice that $x_{1}x_{2}=\frac{1}{b^{q-1}}\in U_{q+1}$ which implies that either both $x_{1}$ and $x_{2}$ are in $U_{q+1}$ or neither $x_{1}$ nor $x_{2}$ is in $U_{q+1}$. Thus it suffices to consider whether $x_{1}$ is in $U_{q+1}$ or not.

Let $\theta=\frac{b^{q+1}-4}{b^{q+1}}\in \fq^*$. Then we have $x_{1}=b\frac{-1+\sqrt{\theta}}{2}$ and
\begin{eqnarray*}
x_{1}^{q+1}=\frac{b^{q+1}}{4}(\sqrt{\theta}^{q+1}-\sqrt{\theta}^q-\sqrt{\theta}+1).
\end{eqnarray*}
This together with the identities $b^{q+1}\sqrt{\theta}^{q+1}=(b^{q+1}-4)\sqrt{\theta}^{q-1}$, $b^{q+1}\sqrt{\theta}^q=\frac{(b^{q+1}-4)}{\sqrt{\theta}}\sqrt{\theta}^{q-1}$ and $b^{q+1}\sqrt{\theta}=\frac{(b^{q+1}-4)}{\sqrt{\theta}^q}\sqrt{\theta}^{q-1}$ implies that
\begin{eqnarray}\label{lem2-eq-1}
x_{1}^{q+1}=\frac{1}{4}((b^{q+1}-4)(1-\frac{1}{\sqrt{\theta}}-\frac{1}{\sqrt{\theta}^q})\sqrt{\theta}^{q-1}+b^{q+1}).
\end{eqnarray}
Notice that $\sqrt{\theta}^{q-1}=\theta^{\frac{q-1}{2}}$ is equal to $1$ (resp. $-1$) if $\theta$ is a square (resp. nonsquare) of $\fq^*$. We then consider the following two cases.

\textbf{Case 1}: $\theta$ is a square of $\fq^*$, i.e., $\sqrt{\theta}^{q-1}=1$. For this case, by \eqref{lem2-eq-1}, one gets
\begin{eqnarray}\label{lem2-eq-2}
x_{1}^{q+1}=\frac{1}{2}(b^{q+1}-2-\frac{b^{q+1}-4}{\sqrt{\theta}}).
\end{eqnarray}
Suppose that $x_{1}^{q+1}=1$, i.e., $x_1\in U_{q+1}$. Then by \eqref{lem2-eq-2} one obtains $\sqrt{\theta}=1$, i.e., $\theta=1$, which leads to $b^{q+1}-4=b^{q+1}$, a contradiction. This proves that $x_1\notin U_{q+1}$.

\textbf{Case 2}: $\theta$ is a nonsquare of $\fq^*$, i.e, $\sqrt{\theta}^{q-1}=-1$. It can be readily verified that in this case \eqref{lem2-eq-1} becomes $x_{1}^{q+1}=1$, i.e., $x_1\in U_{q+1}$.

This completes the proof.
\end{proof}

\section{Main results}
In this section, we investigate the differential spectrum and boomerang spectrum of the power mapping $F(x)=x^{k(q-1)}$ over $\fqtwo$, where $q=p^m$, $p$ is a prime and $m$, $k$ are positive integers with $\gcd(k,q+1)=1$.

\subsection{The differential spectrum and boomerang spectrum of $x^{k(p^m-1)}$ for $p=2$}

\begin{thm}\label{differential-p=2}
Let $F(x)=x^{k(q-1)}$ be a power mapping over $\fqtwo$, where $q=2^m$ and $m$, $k$ are positive integers with $\gcd(k,q+1)=1$. Then $F(x)$ is locally-APN with the differential spectrum
\[\mathbb{DS}_{F}=\{\omega_{0}=2^{2m-1}+2^{m-1}-2, \omega_{2}=2^{2m-1}-2^{m-1}+1, \omega_{2^m-2}=1\}\]
if $m$ is even; and for odd $m$, its differential spectrum is given by
\[\mathbb{DS}_{F}=\{\omega_{0}=2^{2m-1}+2^{m-1}-1, \omega_{2}=2^{2m-1}-2^{m-1}-1, \omega_{4}=1, \omega_{2^m-2}=1\}.\]
\end{thm}
\begin{proof}
For $b\in \fqtwo$, the derivative equation $\mathbb{D}_{1}F(x)=b$ is
\begin{eqnarray}\label{thm1-eq-1}
(x+1)^{k(q-1)}+x^{k(q-1)}=b.
\end{eqnarray}
Observe that $x=0$ and $x=1$ are the solutions of \eqref{thm1-eq-1} if and only if $b=1$. Thus we only need to consider $x \in \fqtwo \backslash \ftwo$ for \eqref{thm1-eq-1} in the following.

\textbf{Case 1}: $b=0$. In this case \eqref{thm1-eq-1} reduces to
\begin{eqnarray}\label{thm1-eq-2}
(1+x^{-1})^{k(q-1)}=1.
\end{eqnarray}
Note that $1+x^{-1}\not=0$ for $x \in \fqtwo \backslash \ftwo$ which implies $(1+x^{-1})^{q-1}\in U_{q+1}$. Thus by \eqref{thm1-eq-2} one has
\begin{eqnarray*}
(1+x^{-1})^{q-1}=1
\end{eqnarray*}
due to $\gcd(k,q+1)=1$.
This indicates that $\fq \backslash \ftwo$ is the solution set of \eqref{thm1-eq-2}.
That is to say, $\delta_F(1,0)=q-2$.

\textbf{Case 2}: $b\ne 0$. Multiplying $x^{k}(x+1)^{k}$ on both sides of \eqref{thm1-eq-1} gives
\begin{eqnarray}\label{thm1-eq-4}
(x+1)^{kq}x^{k}+x^{kq}(x+1)^{k}=bx^{k}(x+1)^{k}.
\end{eqnarray}
Taking $q$-th power on both sides of \eqref{thm1-eq-4} yields
\begin{eqnarray}\label{thm1-eq-5}
(x+1)^{k}x^{kq}+x^{k}(x+1)^{kq}=(bx^{k}(x+1)^{k})^{q}.
\end{eqnarray}
By \eqref{thm1-eq-4} and \eqref{thm1-eq-5}, one then obtains
\begin{eqnarray}\label{thm1-eq-9}
(bx^{k}(x+1)^{k})^{q-1}=1
\end{eqnarray}
which indicates
\begin{eqnarray}\label{thm1-eq-6}
(x+1)^{k(q-1)}= \frac{1}{b^{q-1}x^{k(q-1)}}.
\end{eqnarray}
Plugging \eqref{thm1-eq-6} into \eqref{thm1-eq-1} leads to
\begin{eqnarray}\label{thm1-eq-7}
\frac{1}{b^{q-1}x^{k(q-1)}}+x^{k(q-1)}=b.
\end{eqnarray}
By letting $y=x^{k(q-1)}\in U_{q+1}$ and multiplying $y$ on both sides of \eqref{thm1-eq-7}, one can further obtain
\begin{eqnarray}\label{thm1-eq-8}
y^2+by+\frac{1}{b^{q-1}}=0.
\end{eqnarray}
Clearly, \eqref{thm1-eq-8} has $0$ or $2$ solutions in $\fqtwo$ for $b\in \fqtwo^*$. Moreover, if \eqref{thm1-eq-8} has two solutions $y_{1}, y_{2} \in \fqtwo$, then $y_{1}\in U_{q+1}$ if and only if $y_{2}\in U_{q+1}$ due to $y_{1}y_{2}=\frac{1}{b^{q-1}} \in U_{q+1}$
Thus, by Lemma \ref{square-eq-p=2}, we conclude that \eqref{thm1-eq-8} has two solutions in $U_{q+1}$ if and only if $\Tr_{1}^{m}(\frac{1}{b^{q+1}})=1$.

Now suppose that $\Tr_{1}^{m}(\frac{1}{b^{q+1}})=1$ and let $y$ be a solution of \eqref{thm1-eq-8} in $U_{q+1}$. Since
$x^{k}$ permutes $U_{q+1}$, there exists a unique $z\in U_{q+1}$ such that $z^k=y$ and  a unique $b'\in U_{q+1}$ such that $b^{q-1}=b'^{k}$ respectively. Recalling $y=x^{k(q-1)}\in U_{q+1}$, one has $x^{q-1}=z$. Then from \eqref{thm1-eq-9} we have $(b'(x(x+1))^{q-1})^{k}=1$ and consequently $b'(x(x+1))^{q-1}=1$, which leads to $b' z(zx+1)=x+1$ due to $x^{q-1}=z$. Thus we have
$(b'z^2+1)x=b'z+1$. Now we claim that $b'z^2\ne 1$. Suppose that $b'z^2 = 1$. Then $b'^{k}z^{2k}=b^{q-1}y^{2}=1$. This together with \eqref{thm1-eq-8} gives $by=0$, a contradiction. Hence $x=\frac{b'z+1}{b'z^2+1}$.
Therefore, each solution $y\in U_{q+1}$ of \eqref{thm1-eq-8} provides at most one solution to \eqref{thm1-eq-1}, which means that \eqref{thm1-eq-1} has at most two solutions in $\fqtwo \backslash \ftwo$ for $b\ne 0$. Recall that $x\in \ftwo$ is the solution of \eqref{thm1-eq-1} only if $b=1$. Thus we have $\delta_F(1,b)\leq 2$ for $b\in \fqtwo \backslash \ftwo$, which implies that $F(x)$ is locally-APN.

To compute $\delta_F(1,1)$, we now consider whether \eqref{thm1-eq-8} provides two solutions to \eqref{thm1-eq-1} when $b=1$ or not. If $m$ is even, i.e., $\Tr_{1}^{m}(1)=0$, then \eqref{thm1-eq-8} has no solution in $U_{q+1}$ which means that \eqref{thm1-eq-1} only has the two solutions $x=0,1$, namely, $\delta_F(1,1)=2$.
If $m$ is odd, one can check that $w$ and $w^2$ are the solutions of \eqref{thm1-eq-1} due to $q\equiv -1 \pmod{3}$ and $\gcd(q+1,k)=1$, where $w$ is a primitive $3$rd root of unity over $\fqtwo$. That is to say, $\delta_F(1,1)=4$ if $m$ is odd.

With the above discussion, we can conclude that $\delta_F(1,0)=q-2$, $\delta_F(1,1) = 2$ (resp. $\delta_F(1,1) = 4$) and $\delta_F(1,b) \leq 2$ for $b\in \fqtwo \backslash \ftwo$ if $m$ is even (resp. $m$ is odd).
This together with \eqref{identity-differential} gives the differential spectrum of $F(x)$.
\end{proof}

The following result is useful to compute the boomerang uniformity of $F(x)=x^{k(q-1)}$.
\begin{cor} \label{diff-p=2-lem1}
Let $F(x)=x^{k(q-1)}$ be a power mapping over $\fqtwo$, where $q=2^m$ and $m$, $k$ are positive integers with $\gcd(k,q+1)=1$. Then $\delta_F(1,w)=\delta_F(1,w^2)=0$, where $w$ is a primitive $3$rd root of unity over $\fqtwo$.
\end{cor}
\begin{proof}
Observe that $y=1$ is a solution of \eqref{thm1-eq-8} when $b\in \{w, w^2\}$. According to the proof of Theorem \ref{differential-p=2}, $y=1$ may provide the solution $x=\frac{b'z+1}{b'z^2+1}$ to \eqref{thm1-eq-1}, where $b'$ and $z$ are defined as in the proof of Theorem \ref{differential-p=2}. Clearly, $z=1$ for $y=1$, which indicates $x=1$. However, $x=1$ is not the solution of \eqref{thm1-eq-1} for $b=w \mbox{ or } w^2$, namely, the solution $y=1$ of \eqref{thm1-eq-8} does not contribute solutions to \eqref{thm1-eq-1}. Thus we have $\delta_F(1,b)<2$ for $b\in \{w, w^2\}$ by the proof of Theorem \ref{differential-p=2}, and consequently $\delta_F(1,w)=\delta_F(1,w^2)=0$ since $\delta_F(1,b)$ is even for any $b\in \fqtwo$ in the case $p=2$.
\end{proof}

The boomerang spectrum of $F(x)$ can be determined from its differential spectrum as below.

\begin{thm}\label{boomerang-p=2}
Let $F(x)=x^{k(q-1)}$ be a power mapping over $\fqtwo$, where $q=2^m$ and $m$, $k$ are positive integers with $\gcd(k,q+1)=1$. Then the boomerang spectrum of $F(x)$ is
\[\mathbb{BS}_{F}=\{\nu_{0}=2^{2m-1}+2^{m-1}-2, \nu_{2}=2^{2m-1}-2^{m-1}+1\}\]
if $m$ is even; and for odd $m$, its boomerang spectrum is given by
\[\mathbb{BS}_{F}=\{\nu_{0}=2^{2m-1}+2^{m-1}-3, \nu_{2}=2^{2m-1}-2^{m-1}-1, \nu_{4}=3 \}.\]
\end{thm}

\begin{proof}
According to the definition of the boomerang spectrum of $F(x)$, we need to compute the number $\beta_{F}(1,b)$ of solutions $(x,y)$ of the system of equations
\begin{eqnarray} \label{thm3-eq-1}
\left\{\begin{array}{ll}
x^{k(q-1)}-y^{k(q-1)}=b,  \\
(x+1)^{k(q-1)}-(y+1)^{k(q-1)}=b
\end{array} \right.
\end{eqnarray}
as $b$ runs through $\fqtwo^*$. Clearly, $x\ne y$ in \eqref{thm3-eq-1}.
Note that \eqref{thm3-eq-1} is equivalent to
\begin{eqnarray}\label{thm3-eq-2}
\left\{\begin{array}{ll}
x^{k(q-1)}-y^{k(q-1)}=b,\\
\mathbb{D}_{1}(x)=\mathbb{D}_{1}(y),
\end{array} \right.
\end{eqnarray}
where $\mathbb{D}_{1}(x)=(x+1)^{k(q-1)}-x^{k(q-1)}$.

Let $\mathbb{D}_{1}(x)=\mathbb{D}_{1}(y)=c$ for some $c\in \fqtwo$, we consider \eqref{thm3-eq-2} as follows:

\textbf{Case 1}: $m$ is even. In this case, according to Theorem \ref{differential-p=2}, $\mathbb{D}_{1}(x)=c$ has either $0$, $2$ or $q-2$ solutions as $c$ runs through $\fqtwo$. For $i=0, 2, q-2$, define
\[\Omega_{i}=\{c\in \fqtwo: \delta_F(1,c) = i\}.\]

\textbf{Case 1.1}: $c\in \Omega_{0}$. That is, $\mathbb{D}_{1}(x)=c$ has no solution in $\fqtwo$. Hence \eqref{thm3-eq-2} has no solution for any $b\in \fqtwo^*$.

\textbf{Case 1.2}: $c\in \Omega_{q-2}$. This case happens only if $c=0$. According to Case 1 of the proof of Theorem \ref{differential-p=2}, $\mathbb{D}_{1}(x)=\mathbb{D}_{1}(y)=c$ holds for $x,y \in \fq\backslash \ftwo$. This leads to $x^{k(q-1)}-y^{k(q-1)}=0$. Hence \eqref{thm3-eq-2} has no solution for any $b\in \fqtwo^*$.

\textbf{Case 1.3}: $c\in \Omega_{2}$. For this case, $\mathbb{D}_{1}(x)=c$ has exactly two solutions $x_{0}$ and $x_{0}+1$ in $\fqtwo$ for a fixed $c\in \Omega_{2}$. Therefore, \eqref{thm3-eq-2} has exactly two solutions $(x_{0},x_{0}+1)$ and $(x_{0}+1, x_{0})$ for $b=c\in \Omega_{2}$.

Combining Cases 1.1, 1.2 and 1.3, we conclude that \eqref{thm3-eq-2} has $0$ or $2$ solutions for any $b\in \fqtwo^*$, and $\nu_{2}=|\Omega_{2}|=\omega_{2}=2^{2m-1}-2^{m-1}+1$ by Theorem \ref{differential-p=2}. Then the boomerang spectrum of $F(x)$ follows from $\nu_{0}+\nu_{2}=2^{2m}-1$.

\textbf{Case 2}: $m$ is odd. According to Theorem \ref{differential-p=2}, when $m$ is odd, $\mathbb{D}_{1}(x)=c$ has either $0$, $2$, $4$ or $q-2$ solutions. For $i=0, 2, 4, q-2$, define
\[\Omega_{i}=\{c\in \fqtwo: \delta_F(1,c) = i\}.\]

\textbf{Case 2.1}: $c\in \Omega_{0}$ or $c\in \Omega_{q-2}$. Similar to the discussion in Cases 1.1 and 1.2, in this case \eqref{thm3-eq-2} has no solution for any $b\in \fqtwo^*$.

\textbf{Case 2.2}: $c\in \Omega_{2}$. When this case occurs, $\mathbb{D}_{1}(x)=c$ has exactly two solutions $x_{0}$ and $x_{0}+1$ in $\fqtwo$ for a fixed $c\in \Omega_{2}$. Consequently, this provides exactly two solution $(x_{0},x_{0}+1)$ and $(x_{0}+1, x_{0})$ to \eqref{thm3-eq-2} for $b=c\in \Omega_{2}$.

\textbf{Case 2.3}: $c\in \Omega_{4}$. That is, $c=1$. In this case, according to the proof of Theorem \ref{differential-p=2}, $\{0,1,w,w^2\}$ is the solution set of $\mathbb{D}_{1}(x)=c$, where $w$ is a primitive $3$rd root of unity over $\fqtwo$. Note that $(x,y)$ is a solution of
\eqref{thm3-eq-2} for $b=x^{k(q-1)}-y^{k(q-1)}$ when $(x,y)\in \{(x,y): x,y \in \{0,1,w,w^2\} \mbox{ and } x\ne y\}$. Then it can be readily verified that this provides four solutions $(x,y)$ to \eqref{thm3-eq-2} for each $b\in \{1,w,w^2\}$.

With the above discussion and the fact that $1,w,w^2 \notin \Omega_{2}$ by Corollary \ref{diff-p=2-lem1}, we conclude that $\beta_{F}(1,b)$ equals $0$, $2$ or $4$. Further, we have $\nu_{2}=|\Omega_{2}|=\omega_{2}=2^{2m-1}-2^{m-1}-1$ and  $\nu_{4}=3$.
Then  the boomerang spectrum of $F(x)$ follows from $\nu_{0}+\nu_{2}+\nu_{4}=2^{2m}-1$.

This completes the proof.
\end{proof}

\subsection{The differential spectrum and boomerang spectrum of $x^{k(p^m-1)}$ for $p>2$}

\begin{thm}\label{differential-p>2}
Let $F(x)=x^{k(q-1)}$ be a power mapping over $\fqtwo$, where $q=p^m$, $p$ is an odd prime and $m$, $k$ are positive integers with $\gcd(k,q+1)=1$.
Then $F(x)$ is locally-APN with the differential spectrum
\[\mathbb{DS}_{F}=\{\omega_{0}=\frac{q^2-1}{2}-(q-1), \omega_{1}=3(q-1), \omega_{2}=\frac{q^2-1}{2}-2q, \omega_{3}=2, \omega_{q-2}=1\}\]
if $q\equiv 2\pmod{3}$; and otherwise
\[\mathbb{DS}_{F}=\{\omega_{0}=\frac{q^2-1}{2}-(q+1), \omega_{1}=3q-1, \omega_{2}=\frac{q^2-1}{2}-2(q-1), \omega_{q-2}=1\}.\]
\end{thm}

\begin{proof}
For $b\in \fqtwo$, the derivative equation $\mathbb{D}_{1}F(x)=b$ is
\begin{eqnarray}\label{thm2-eq-1}
(x+1)^{k(q-1)}-x^{k(q-1)}=b.
\end{eqnarray}
Observe that $x=0$ (resp. $x=-1$) is a solution of \eqref{thm2-eq-1} if and only if $b=1$ (resp. $b=-1$). Thus it suffices to consider $x \in \fqtwo \backslash\{0,-1\}$ for \eqref{thm2-eq-1}.

\textbf{Case 1}: $b=0$. In this case, \eqref{thm2-eq-1} reduces to
\begin{eqnarray}\label{thm2-eq-2}
(1+x^{-1})^{k(q-1)}=1.
\end{eqnarray}
Note that $1+x^{-1}\not=0$ for $x \in \fqtwo \backslash \{0,-1\}$. Thus by \eqref{thm2-eq-2} one has
$(1+x^{-1})^{q-1}=1$
since $x^{k}$ permutes $U_{q+1}$ due to $\gcd(k,q+1)=1$.
Then one can verify that $\fq \backslash \{0,-1\}$ is the solution set of \eqref{thm2-eq-2}.
That is to say, $\delta_F(1,0)=q-2$.

\textbf{Case 2}: $b\ne 0$. Notice that $x(x+1)\ne 0$ for $x \in \fqtwo \backslash \{0,-1\}$. Multiplying $x^k(x+1)^k$ on both sides of \eqref{thm2-eq-1} gives
\begin{eqnarray}\label{thm2-eq-3}
(x+1)^{kq}x^{k}-x^{kq}(x+1)^{k}=bx^{k}(x+1)^{k}.
\end{eqnarray}
Taking $q$-th power on both sides of \eqref{thm2-eq-3} yields
\begin{eqnarray}\label{thm2-eq-4}
(x+1)^{k}x^{kq}-x^{k}(x+1)^{kq}=(bx^{k}(x+1)^{k})^{q}.
\end{eqnarray}
Then, by \eqref{thm2-eq-3} and \eqref{thm2-eq-4}, one obtains
\begin{eqnarray}\label{thm2-eq-5}
(bx^{k}(x+1)^{k})^{q-1}=-1
\end{eqnarray}
which implies
\begin{eqnarray}\label{thm2-eq-6}
(x+1)^{k(q-1)}=-\frac{1}{(bx^{k})^{q-1}}.
\end{eqnarray}
Eliminating the term $(x+1)^{k(q-1)}$ from \eqref{thm2-eq-1} and \eqref{thm2-eq-6} leads to
\begin{eqnarray}\label{thm2-eq-12}
\frac{1}{(bx^{k})^{q-1}}+x^{k(q-1)}+b=0.
\end{eqnarray}
By letting $y=x^{k(q-1)} \in U_{q+1}$ and multiplying $y$ on both sides of \eqref{thm2-eq-12}, one gets
\begin{eqnarray}\label{thm2-eq-7}
y^2+by+\frac{1}{b^{q-1}}=0.
\end{eqnarray}

Let $\Delta=\frac{b^{q+1}-4}{b^{q-1}}$ be the discriminant of the quadratic equation \eqref{thm2-eq-7}. Clearly, \eqref{thm2-eq-7} has the unique solution $y=-b/2$ in $\fqtwo$ only if $\Delta=0$, i.e., $b^{q+1}=4$. Note that $-b/2\in U_{q+1}$ if $b^{q+1}=4$. When $\Delta\ne0$, then it is a square of $\fqtwo^*$ due to $b^{q+1}-4\in\fq$, and thus \eqref{thm2-eq-7} has two solutions $y_{1}$ and $y_{2}$ in $\fqtwo$. Observe that $y_{1}y_{2}=\frac{1}{b^{q-1}} \in U_{q+1}$ which implies either both $y_{1}$ and $y_{2}$ are in $U_{q+1}$ or neither $y_{1}$ nor $y_{2}$ is in $U_{q+1}$. Hence \eqref{thm2-eq-7} has a unique solution in $U_{q+1}$ if and only if $b^{q+1}=4$ and the solution is $y=-b/2$. By Lemma \ref{square-eq-p>2}, \eqref{thm2-eq-7} has two solutions in $U_{q+1}$ if and only if $\theta:=\frac{b^{q+1}-4}{b^{q+1}}$ is a nonsquare of $\fq^*$. Further, its two solutions are $y_{1}=b\frac{-1+\sqrt{\theta}}{2}$ and $y_{2}=b\frac{-1-\sqrt{\theta}}{2}$.

Suppose that $y \in U_{q+1}$ is a solution of \eqref{thm2-eq-7}. Then there exists a unique $z\in U_{q+1}$ such that $z^k=y$ since $x^k$ permutes $U_{q+1}$ due to $\gcd(k,q+1)=1$. Thus one has $x^{q-1}=z$ due to $y=x^{k(q-1)}$. Let $b'$ be the unique element in $U_{q+1}$ such that $b^{q-1}=b'^k$. Then by \eqref{thm2-eq-5}, it gives
\begin{eqnarray}\label{thm2-eq-8}
b'z(x+1)^{q-1}=-1
\end{eqnarray}
since $x^k$ permutes $U_{q+1}$ and $k$ is odd due to $\gcd(k,q+1)=1$. It follows from \eqref{thm2-eq-8} and $x^{q-1}=z$ that \begin{eqnarray}\label{thm2-eq-9}
x=-\frac{b'z+1}{b'z^2+1}.
\end{eqnarray}
Here we claim that $b'z^2+1 \ne 0$. Suppose that $b'z^2 = -1$. Taking $k$-th power on both sides of this equation gives $b^{q-1}y^2=-1$. By \eqref{thm2-eq-7}, it leads to $b^qy=0$, a contradiction.
This shows that a solution $y \in U_{q+1}$ of \eqref{thm2-eq-7} provides at most one solution to \eqref{thm2-eq-1}. Therefore \eqref{thm2-eq-1} has at most two solutions in $\fqtwo \backslash \{0,-1\}$ for $b\ne 0$. Recall that $x=0$ (resp. $x=-1$) is the solution of \eqref{thm2-eq-1} only if $b=1$ (resp. $b=-1$). Thus $\delta_F(1,b)\leq 2$ for $b\in \fqtwo \backslash \{0,\pm 1\}$, which implies that $F(x)$ is locally-APN.

Next we discuss the solutions of \eqref{thm2-eq-1} in detail in order to obtain the differential spectrum of $F(x)$. To do this, we first claim the following two statements:\\
\textbf{Statement 1)}. When $b^{q+1} = 4$, then \eqref{thm2-eq-1} has exactly one solution $x=-\frac{b'z+1}{b'z^2+1}$ in $\fqtwo \backslash \{0,-1\}$ if $b\ne \pm 2$ and has no solution in $\fqtwo \backslash \{0,-1\}$ if $b = \pm 2$, where $b^{q-1}=b'^k$ and $z^k=y=-b/2$.\\
\textbf{Statement 2)}. If $\theta=\frac{b^{q+1}-4}{b^{q+1}}$ is a nonsquare of $\fq^*$, then \eqref{thm2-eq-1} has at least one solution $x_{1}=-\frac{b'z_{1}+1}{b'z_{1}^2+1}$ or $x_{2}=-\frac{b'z_{2}+1}{b'z_{2}^2+1}$ in $\fqtwo \backslash \{0,-1\}$, where $b^{q-1}=b'^k$, $z_{1}^k=y_{1}=b\frac{-1+\sqrt{\theta}}{2} $ and $z_{2}^k=y_{2}=b\frac{-1-\sqrt{\theta}}{2} $.

Statements 1) and 2) can be verified as follows. When $b^{q+1} = 4$, one has that $b'z+1=0$ if and only if $b=2$ and $b'z+1 = b'z^2+1$ if and only if $z=1$, i.e., $b=-2$. Thus \eqref{thm2-eq-1} has no solution in  $ \fqtwo \backslash \{0,-1\}$ if $b=\pm 2$.
Recall that $b'z^2+1 \ne 0$. Hence $(z-1)(b'z+1)(b'z^2+1)\ne 0$ if $b\ne \pm 2$. Now we focus on the case $b \ne \pm 2$ when $b^{q+1} = 4$. Using $x=-\frac{b'z+1}{b'z^2+1}$, a direct computation gives
\begin{eqnarray}\label{thm2-eq-10}
x^{q-1}=\frac{(b'z+1)(b'z^2+1)}{(b'z+1)(b'z^2+1)}z=z
\end{eqnarray}
and
\begin{eqnarray}\label{thm2-eq-11}
b'z(x+1)^{q-1}=-\frac{(z-1)(b'z^2+1)}{(z-1)(b'z^2+1)}=-1.
\end{eqnarray}
Then by \eqref{thm2-eq-10} and \eqref{thm2-eq-11} we have
\[(x+1)^{k(q-1)}-x^{k(q-1)}=(-\frac{1}{b'z})^{k}-z^k=-\frac{1}{b^{q-1}y}-y=b\]
due to $y=-b/2$, $b^{q-1}=b'^k$ and $b^{q+1} = 4$. Thus $x=-\frac{b'z+1}{b'z^2+1}$ is a solution of \eqref{thm2-eq-1} in  $ \fqtwo \backslash \{0,-1\}$ if $b^{q+1} = 4$ and $b \ne \pm 2$, and it is the unique solution. This proves statement 1).

For the second statement, we first show that at least one of $x_{1}=-\frac{b'z_{1}+1}{b'z_{1}^2+1}$ and $x_{2}=-\frac{b'z_{2}+1}{b'z_{2}^2+1}$ is not in $\{0,-1\}$. Without loss of generality, we only need to prove that $x_{2}\notin \{0,-1\}$ if $x_{1}\in \{0,-1\}$.
Notice that $x_{i}=0$, i.e., $b'z_{i}=-1$, implies $(b'z_{i})^{k}=b^{q-1}y_{i}=-1$ for $i\in\{1,2\}$.
If $x_{1}=x_{2}=0$, then $z_{1}=z_{2}$, a contradiction. Suppose that $x_{1} =0$, $x_{2} = -1$. Then we have $z_{2}=1$ and thus $y_{2}=1$, which implies $y_{1}=\frac{1}{b^{q-1}}$ due to $y_{1}y_{2}=\frac{1}{b^{q-1}}$. This means that $b^{q-1}y_{1}=1$, which is a contradiction with $b^{q-1}y_{1}=-1$ due to $x_{1} =0$. This proves that $x_{2}\ne 0,-1$ if $x_{1}=0$. Similarly, we can prove that $x_{2}\ne 0,-1$ if $x_{1}=-1$. Thus at least one of $x_{1}$ and $x_{2}$ is not in $\{0,-1\}$.
Without loss of generality, assume that $x_{1}\notin \{0,-1\}$, which implies $(z_{1}-1)(b'z_{1}+1)\ne 0$. Similar to the computation in the proof of statement 1), a direct calculation gives
\[(x_{1}+1)^{k(q-1)}-x_{1}^{k(q-1)}=-\frac{1}{b^{q-1}y_{1}}-y_{1}=b\]
due to $x_{1}=-\frac{b'z_{1}+1}{b'z_{1}^2+1}$ and $y_{1}=b\frac{-1+\sqrt{\theta}}{2}$. Thus $x_{1}$ is a solution of \eqref{thm2-eq-1}. This proves statement 2).

Combining Cases 1 and 2 we conclude that $\delta_F(1,0)=q-2$, $\delta_F(1,1)\leq 3$, $\delta_F(1,-1) \leq 3$ and $\delta_F(1,b)\leq 2$ for $b\in \fqtwo \backslash \{0,\pm 1\}$, which implies $\omega_{q-2}=1$, and $\delta_F(1,b)>0$ if and only if 1) $b=0$; or 2) $b=\pm 1$; or 3) $b^{q+1}=4$ and $b\ne \pm 2$; or 4) $\theta=\frac{b^{q+1}-4}{b^{q+1}}$ is a nonsquare of $\fq^*$. Notice that  $\theta=-3$ if $b=\pm 1$, where $-3$ is a nonsquare of $\fq^*$ if and only if $q\equiv 2\pmod{3}$ by Lemma \ref{Issquare(-3)}. Thus it gives
\begin{eqnarray*}
\omega_{0}=
\left\{\begin{array}{ll}
\frac{q^2-1}{2}-(q-1), & \mbox{if $q\equiv 2\pmod{3}$}; \\
\frac{q^2-1}{2}-(q+1),&\mbox{otherwise}.
\end{array} \right.
\end{eqnarray*}

Now we determine the values of $\delta_F(1,1)$ and $\delta_F(1,-1)$. Note that $\delta_F(1,1)=\delta_F(1,-1)$ since $\mathbb{D}_{1}F(x)=b$ has a solution $x$ if and only if $\mathbb{D}_{1}F(x)=-b$ has a solution $-x-1$, which indicates $\delta_F(1,b)=\delta_F(1,-b)$ for any $b\in \fqtwo^*$. Thus we only need to consider $\delta_F(1,1)$ in the following. Note that $x=0$ is always a solution of \eqref{thm2-eq-1} for $b=1$, i.e., $\delta_F(1,1)\geq 1$. If $p=3$, $b^{q+1}=4$ for $b=1$ and thus $\delta_F(1,1)= 1$ due to the statement 1) in Case 2. If $q \not\equiv 2 \pmod{3}$ and $p\ne3$, one has $\theta=-3\ne 0$ for $b=1$, which is a square of $\fq^*$ by Lemma \ref{Issquare(-3)}, and thus $\delta_F(1,1)= 1$ due to the discussion in Case 2. If $q \equiv 2 \pmod{3}$, it can be verified that $-w$ and $-w^2$ are also the solutions of \eqref{thm2-eq-1} for $b=1$ due to the facts that $(w-1)^{q-1}=-w^2$ and $(w^2-1)^{q-1}=-w$.
Thus $\delta_F(1,1)=3$ if $q\equiv 2 \pmod{3}$.
Therefore we conclude that
\begin{eqnarray*}
\delta_F(1,1)=\delta_F(1,-1)=
\left\{\begin{array}{ll}
3, & \mbox{if $q\equiv 2\pmod{3}$}; \\
1,&\mbox{otherwise},
\end{array} \right.
\end{eqnarray*}
which indicates
\begin{eqnarray*}
\omega_{3}=
\left\{\begin{array}{ll}
2, & \mbox{if $q\equiv 2\pmod{3}$}; \\
0,&\mbox{otherwise}.
\end{array} \right.
\end{eqnarray*}

Then the differential spectrum of $F(x)$ follows from \eqref{identity-differential}. This completes the proof.
\end{proof}

The following corollary is useful to compute the boomerang spectrum of $F(x)=x^{k(q-1)}$.

\begin{cor} \label{diff-p>2-lem1}
Let $m$ and $k$ be positive integers, $p$ be an odd prime and $q=p^m$. Assume that $\gcd(k,q+1)=1$ and $q\equiv 2\pmod{3}$. Let $F(x)=x^{k(q-1)}$ be a power mapping over $\fqtwo$. If $b^{q+1}-4=w^2 b^{q-1}$, then $\delta_F(1,b)\ne 2$, where $w$ is a primitive $3$rd root of unity over $\fqtwo$.
\end{cor}
\begin{proof}
Notice that $b^{q+1}-4 \in \fq$ and $b^{q-1}w^2 \in U_{q+1}$ for $q\equiv 2\pmod{3}$. Thus $b^{q+1}-4=b^{q-1}w^2= 1$ or $-1$ due to $\fq \cap U_{q+1} =\{1,-1\}$. Then we consider the following two cases:

\textbf{Case 1}: $b^{q+1}-4=b^{q-1}w^2= 1$. Clearly, we have $b^{q+1}=5$, $b^{q-1}=w $ and then $b^2=5w^2$ in this case. Due to $b^2=5w^2$ and $b^{q+1}=5$, we have $b^{q+1}=(5w^2)^{\frac{q+1}{2}}=5^{\frac{q+1}{2}}=5^{\frac{q-1}{2}}5=5$, which implies $5^{\frac{q-1}{2}}=1$, namely, $5$ is a square of $\fq^*$. Thus $\theta=\frac{b^{q+1}-4}{b^{q+1}}=\frac{1}{5}$ is also a square of $\fq^*$. This together with the fact that $b\notin \{0,\pm 1\}$ means that $\delta_F(1,b)\ne 2$ in this case according to the proof of Theorem \ref{differential-p>2}.

\textbf{Case 2}: $b^{q+1}-4=b^{q-1}w^2=-1$. If this case occurs, one then has $b^{q+1}=3$, $b^{q-1}=-w$ and $b^2=-3w^2$. Notice that $\{1,\frac{-\sqrt{-3}-1}{2}, \frac{\sqrt{-3}-1}{2}\}$ is the solution set of $x^3 = 1$ over $\fqtwo$. Without loss of generality, let $w=\frac{-\sqrt{-3}-1}{2}$ and $w^2=\frac{\sqrt{-3}-1}{2}$. Note that $b^2=-3w^2$ implies $b=\pm \sqrt{-3}w$. Next we only focus on the case $b=\sqrt{-3}w$ since $\delta_F(1,-\sqrt{-3}w)=\delta_F(1,\sqrt{-3}w)$.
Recall from the proof of Theorem \ref{differential-p>2} that $\delta_F(1,\sqrt{-3}w) \leq 2$ due to the fact that $b=\sqrt{-3}w \notin \{0,\pm 1\}$ for $q\equiv 2\pmod{3}$, and each solution $x\in \fqtwo$ of \eqref{thm2-eq-1} for $b=\sqrt{-3}w$ corresponds to a solution $y\in U_{q+1}$ of \eqref{thm2-eq-7}.
It can be verified that $\{-1,w^2\}$ is the solution set of \eqref{thm2-eq-7} for $b=\sqrt{-3}w$. When $b=\sqrt{-3}w$ and $y=w^2$, one has $b^{q-1}y=-1$. This together with the relations $z^k=y$ and $b^{q-1}=b'^k$, one gets $(b'z)^k=-1$, which implies $b'z=-1$. By \eqref{thm2-eq-9}, we have $x=0$. That is to say, the solution $y=w^2$ of \eqref{thm2-eq-7} cannot provide a solution in $\fqtwo \backslash \{0,-1\}$ to \eqref{thm2-eq-1} for $b=\sqrt{-3}w$, namely, $\delta_F(1,\sqrt{-3}w)\ne 2$ by the proof of Theorem \ref{differential-p>2}.

Combining Cases 1 and 2, the desired result follows. This completes the proof.
\end{proof}

\begin{thm}\label{boomerang-p>2}
Let $F(x)=x^{k(q-1)}$ be a power mapping over $\fqtwo$, where $q=p^m$, $p$ is an odd prime and $m$, $k$ are positive integers with $\gcd(k,q+1)=1$. Then the boomerang spectrum of $F(x)$ is given by
\[\mathbb{BS}_{F}=\{\nu_{0}=\frac{q^2-1}{2}+2q-6, \nu_{2}=\frac{q^2-1}{2}-2q+6\}\]
if $q\equiv 2\pmod{3}$; and otherwise
\[\mathbb{BS}_{F}=\{\nu_{0}=\frac{q^2-1}{2}+2(q-1), \nu_{2}=\frac{q^2-1}{2}-2(q-1)\}.\]
\end{thm}

\begin{proof}
According to the definition of the boomerang spectrum of $F(x)$, we need to compute the number $\beta_{F}(1,b)$ of solutions of the system of equations
\begin{eqnarray} \label{thm4-eq-1}
\left\{\begin{array}{ll}
x^{k(q-1)}-y^{k(q-1)}=b,  \\
(x+1)^{k(q-1)}-(y+1)^{k(q-1)}=b
\end{array} \right.
\end{eqnarray}
as $b$ runs through $\fqtwo^*$. Clearly, we have $x\ne y$ and  \eqref{thm4-eq-1} is equivalent to
\begin{eqnarray}\label{thm4-eq-2}
\left\{\begin{array}{ll}
x^{k(q-1)}-y^{k(q-1)}=b,\\
\mathbb{D}_{1}(x)=\mathbb{D}_{1}(y)
\end{array} \right.
\end{eqnarray}
where $\mathbb{D}_{1}(x)=(x+1)^{k(q-1)}-x^{k(q-1)}$.

Let $\mathbb{D}_{1}(x)=\mathbb{D}_{1}(y)=c$ for some $c\in \fqtwo$. Then we can discuss \eqref{thm4-eq-2} as follows:

\textbf{Case 1}: $q \not \equiv 2 \pmod{3}$. In this case, $\mathbb{D}_{1}(x)=c$ has either $0$, $1$, $2$ or $q-2$ solutions as $c$ runs through $\fqtwo$ by Theorem \ref{differential-p>2}. For $i=0, 1, 2, q-2$, define
\[\Omega_{i}=\{c\in \fqtwo: \delta_F(1,c) = i\}.\]

\textbf{Case 1.1}: $c\in \Omega_{0} \mbox{ or } \Omega_{1}$. In this case, it can be readily seen that \eqref{thm4-eq-2} has no solution $(x,y) \in \fqtwo^2$.

\textbf{Case 1.2}: $c\in \Omega_{q-2}$, i.e., $c=0$. According to Case 1 of the proof of Theorem \ref{differential-p>2}, $\fq \backslash \{0,-1\}$ is the solution set of $\mathbb{D}_{1}(x)=0$. This implies that $x,y$ satisfying \eqref{thm4-eq-2} are in $\fq \backslash \{0,-1\}$, which leads to $b=0$ by the first equation of \eqref{thm4-eq-2}.
Hence \eqref{thm4-eq-2} has no solution for any $b\in \fqtwo^*$.

\textbf{Case 1.3}: $c\in \Omega_{2}$. Let $x_{1}$ and $x_{2}$ be exactly the two solutions of $\mathbb{D}_{1}(x)=c$ for a fixed $c\in \Omega_{2}$, and let $b_{0}$ be defined by $x_{1}^{k(q-1)}-x_{2}^{k(q-1)}=b_{0}$. Clearly, $(x_{1},x_{2})$ (resp. $(x_{2},x_{1})$) is a solution of \eqref{thm4-eq-2} for $b=b_{0}$ (resp. $b=-b_{0}$). That is to say, $c\in \Omega_{2}$ contributes exactly one solution $(x_{1},x_{2})$ to \eqref{thm4-eq-2} for $b=b_{0}$.  Note that $-c\in \Omega_{2}$ if $c\in \Omega_{2}$ since $\delta_{F}(1,c)=\delta_{F}(1,-c)$ for any $c\in \fqtwo^*$. Moreover, $-(x_{1}+1)$ and $-(x_{2}+1)$ are exactly the two solutions of $\mathbb{D}_{1}(x)=-c$. This shows that $-c\in \Omega_{2}$ also contributes exactly one solution $(-(x_{1}+1),-(x_{2}+1))$ to \eqref{thm4-eq-2} for $b=b_{0}$.
Now we claim that for any $c'\in \Omega_{2} \backslash \{c,-c\}$ it cannot contribute solutions to \eqref{thm4-eq-2} when $b=b_{0}$.
Suppose that $c'\in \Omega_{2} \backslash \{c,-c\}$ contributes a solution to \eqref{thm4-eq-2} for $b=b_{0}$.
Recall from the proof of Theorem \ref{differential-p>2} that $x_{1}$ and $x_{2}$ correspond to the two solutions $y_{1}$ and $y_{2}$ in $U_{q+1}$ of \eqref{thm2-eq-7} respectively with the relation $x_{i}^{k(q-1)}=y_{i}$ for $i\in \{1,2\}$. Moreover,  without loss of generality, we may assume that $y_{1}=c\frac{-1+\sqrt{\theta}}{2}$ and $y_{2}=c\frac{-1-\sqrt{\theta}}{2}$ by the proof of Theorem \ref{differential-p>2}, where $\theta=\frac{c^{q+1}-4}{c^{q+1}}$ is a nonsquare of $\fq^*$. Then we have
$x_{1}^{k(q-1)}-x_{2}^{k(q-1)}=y_{1}-y_{2}=c\sqrt{\frac{c^{q+1}-4}{c^{q+1}}}$ which gives $c\sqrt{\frac{c^{q+1}-4}{c^{q+1}}}=b_{0}$. This leads to $(\frac{b_{0}}{c})^{q-1}=-1$. Similarly, we also have $c'\sqrt{\frac{c'^{q+1}-4}{c'^{q+1}}}=b_{0}$ and $(\frac{b_{0}}{c'})^{q-1}=-1$. Let $c'=\alpha c$ with $\alpha \in \fqtwo^* \backslash \{1,-1\}$. Clearly, it follows from $(\frac{b_{0}}{c})^{q-1}=(\frac{b_{0}}{c'})^{q-1}=-1$ that $\alpha^{q-1}=1$. This together with $c\sqrt{\frac{c^{q+1}-4}{c^{q+1}}}= c'\sqrt{\frac{c'^{q+1}-4}{c'^{q+1}}}$ gives $\alpha^{2}=1$ due to $c\ne0$, a contradiction with $\alpha \ne 1,-1$. This shows that in this case \eqref{thm4-eq-2} has exactly two solutions $(x_{1},x_{2})$ and $(-(x_{1}+1),-(x_{2}+1))$ for $b=b_{0}$.

Combining Cases 1.1, 1.2 and 1.3, we conclude that \eqref{thm4-eq-2} has either $0$ or $2$ solutions for $b\in \fqtwo^*$, namely, the boomerang uniformity of $F(x)$ is 2. Moreover, $\nu_{2}=|\Omega_{2}|=\omega_{2}=\frac{q^2-1}{2}-2(q-1)$, where the value of $\omega_{2}$ is given by Theorem \ref{differential-p>2}. This together with $\nu_{0}+\nu_{2}=q^2-1$ gives the boomerang spectrum of $F(x)$.

\textbf{Case 2}: $q \equiv 2 \pmod{3}$. If this case happens, then $\mathbb{D}_{1}(x)=c$ has either $0$, $1$, $2$, $3$ or $q-2$ solutions as $c$ runs through $\fqtwo$ by Theorem \ref{differential-p>2}. For $i=0, 1, 2, 3, q-2$, define
\[\Omega_{i}=\{c\in \fqtwo: \delta_F(1,c) = i\}.\]

\textbf{Case 2.1}: $c\in \Omega_{0}, \,\, \Omega_{1} \mbox{ or }  \Omega_{q-2}$. Similar to Cases 1.1 and 1.2, this case contributes no solution $(x,y) \in \fqtwo^2$ to \eqref{thm4-eq-2} for any $b\in \fqtwo^*$.

\textbf{Case 2.2}: $c\in \Omega_{2}$. Let $x_{1}$ and $x_{2}$ be exactly the two solutions of $\mathbb{D}_{1}(x)=c$ for a fixed $c\in \Omega_{2}$, and let $b_{0}$ be defined by $x_{1}^{k(q-1)}-x_{2}^{k(q-1)}=b_{0}$. Denote the set of $b_{0}$'s as $c$ runs through $\Omega_{2}$ by $\mathcal{B}$. Similar to Case 1.3, this case contributes exactly two solutions to \eqref{thm4-eq-2} for each $b \in \mathcal{B}$.
%

\textbf{Case 2.3}: $c\in \Omega_{3}$, i.e., $c=\pm 1$. According to the proof of Theorem \ref{differential-p>2}, $S_{1}=\{0,-w,-w^2\}$ is the solution set of $\mathbb{D}_{1}(x)=1$ and $S_{-1}=\{-1,w-1,w^2-1\}$ is the solution set of $\mathbb{D}_{1}(x)=-1$. Then any $(x,y)$ with $x\ne y$ and $x,y\in S_{i}$ for some $i\in \{1,-1\}$ is a solution of \eqref{thm4-eq-2} for some $b\in \fqtwo^*$. A direct computation shows that this contributes $2$ solutions to \eqref{thm4-eq-2} for each $b\in \{\pm w, \pm w^2, \pm (w^2-w)\}$. In the following we prove that $\pm w, \pm w^2, \pm (w^2-w) \notin \mathcal{B}$:

\textbf{1)}: $\pm (w^2-w) \notin \mathcal{B}$. Clearly, $\gcd(3,k)=1$ since $q \equiv 2 \pmod{3}$ and $\gcd(k,q+1)=1$. Note that $(-w^2,-w)$ and $(w^2-1,w-1)$ are solutions of \eqref{thm4-eq-2} for $b=w^{2k}-w^{k}$, and $(-w,-w^2)$ and $(w-1,w^2-1)$ are solutions of \eqref{thm4-eq-2} for $b=w^{k}-w^{2k}$, where $w^{2k}-w^{k}=w^{2}-w$ if $k\equiv 1 \pmod{3}$ and $w^{2k}-w^{k}=w-w^{2}$ otherwise. Without loss of generality, we only consider the case $k\equiv 1 \pmod{3}$ in the following. Thus, for $c=1$, it contributes the solution $(-w^2,-w)$ (resp. $(-w,-w^2)$) to \eqref{thm4-eq-2} for $b=w^{2}-w$ (resp. $b=w-w^{2}$) and for $c=-1$, it contributes the solution $(w^2-1,w-1)$ (resp. $(w-1,w^2-1)$) to \eqref{thm4-eq-2} for $b=w^{2}-w$ (resp. $b=w-w^{2}$). Recall from the proof of Theorem \ref{differential-p>2} that
the solutions $-w$ and $-w^2$ (resp. $w-1$ and $w^2-1$) of the equation $\mathbb{D}_{1}(x)=c$ for $c=1$ (resp. $c=-1$) correspond to the two solutions in $U_{q+1}$ of \eqref{thm2-eq-7}. Notice that $\pm 1\notin \Omega_{2}$. Thus, similar to the discussion in Case 1.3, one can prove that
for any $c\in \Omega_{2}$ it cannot contribute solutions to \eqref{thm4-eq-2} for $b=\pm (w^2-w)$, which indicates that $\pm (w^2-w) \notin \mathcal{B}$.

\textbf{2)}:  $\pm w, \pm w^2 \notin \mathcal{B}$. We only prove $w\notin \mathcal{B}$ here since the other cases can be proved in the same manner.
Suppose that $b_{0}=w \in \mathcal{B}$ corresponding to some $c\in \Omega_{2}$. Similar to Case 1.3, it follows that $c\sqrt{\frac{c^{q+1}-4}{c^{q+1}}}=b_{0}=w$, which yields $c^{q+1}-4=w^2 c^{q-1}$ by squaring both sides of the equation. By Corollary \ref{diff-p>2-lem1}, we have $c\notin \Omega_{2}$ if $c^{q+1}-4=w^2 c^{q-1}$. Thus $w\notin \mathcal{B}$.

With the discussion in Case 2, we conclude that \eqref{thm4-eq-2} has $0$ or $2$ solutions for $b\in \fqtwo^*$ and $\nu_{2}=|\Omega_{2}|+6=\omega_{2}+6=\frac{q^2-1}{2}-2q+6$. This together with $\nu_{0}+\nu_{2}=q^2-1$ gives the boomerang spectrum of $F(x)$ for the case $q \equiv 2 \pmod{3}$. This completes the proof.
\end{proof}

\section{Conclusion}
In this paper, inspired by an example on a locally-APN function with boomerang uniformity 2 over $\mathbb{F}_{2^8}$ proposed by Hasan, Pal and St\u{a}nic\u{a} in 2021, we studied the power mapping $F(x)=x^{k(q-1)}$ over $\fqtwo$ with $\gcd(k,q+1)=1$ for any prime power $q$ and generalized the example into an infinite class of locally-APN functions with boomerang uniformity 2. We completely determined the differential spectrum and boomerang spectrum of this class of power mappings for any prime power $q$. Our results on the differential spectrum of $F(x)$ extended the result of \cite[Theorem 7]{BCCP} from $(p,k)=(2,1)$ to general $(p,k)$. Moreover, our results on the boomerang spectrum of $F(x)$ generalized the results in \cite{HPSP,YLSF}.

%

\section*{Data Deposition Information}
The data used to support the findings of this study are available from the corresponding author upon request.

\section*{Declaration of competing interest}

The authors have no relevant financial or non-financial interests to disclose.

\section*{Acknowledgments}

This work was supported by the National Key Research and Development Program of China (No. 2021YFA1000600), the National Natural Science Foundation of China (No. 62072162), the Natural Science Foundation of Hubei Province of China (No. 2021CFA079) and the Knowledge Innovation Program of Wuhan-Basic Research (No. 2022010801010319).


\begin{thebibliography}{99}

\bibitem{BSA}  E. Biham, A. Shamir, Differential cryptanalysis of DES-like cryptosystems, J. Cryptology 4(1) (1991), pp. 3--72.

\bibitem{BCCP1} C. Blondeau, A. Canteaut, P. Charpin, Differential properties of power functions, Int. J. Inf. Coding Theory 1(2) (2010), pp. 149--170.

\bibitem{BCCP} C. Blondeau, A. Canteaut, P. Charpin, Differential properties of ${x\mapsto x^{2^{t}-1}}$, IEEE Trans. Inf. Theory 57(12) (2011), pp. 8127--8137.

\bibitem{BPLC} C. Blondeau, L. Perrin, More differentially $6$-uniform power functions, Des. Codes Cryptogr. 73(2) (2014), pp. 487--505.

\bibitem{BCCA}  C. Boura, A. Canteaut, On the boomerang uniformity of cryptographic S-boxes, IACR Trans. Symmetric Cryptol.
    2018(3) (2018), pp. 290--310.

\bibitem{CVIM}  M. Calderini, I. Villa, On the boomerang uniformity of some permutation polynomials, Cryptogr. Commun. 12  (2020), pp. 1161--1178.

\bibitem{CPJP} P. Charpin, J. Peng, Differential uniformity and the associated codes of cryptographic functions, Adv. Math. Commun. 13(4) (2019), pp. 579--600.

\bibitem{CHNC} S.-T. Choi, S. Hong, J.-S. No, H. Chung, Differential spectrum of some power functions in odd prime characteristic, Finite Fields Appl. 21 (2013), pp. 11--29.

\bibitem{CHPS}  C. Cid, T. Huang, T. Peyrin, Y. Sasaki, L. Song, Boomerang Connectivity Table: A new cryptanalysis tool, in  EUROCRYPT 2018, Jesper Buus Nielsen and Vincent Rijmen, editors, LNCS, 10821, pp. 683--714. Springer, Cham, April 2018.

\bibitem{DHKM} H. Dobbertin, T. Helleseth, P. V. Kumar, H. Martinsen, Ternary $m$-sequences with three-valued cross-correlation function: New decimations of Welch and Niho type, IEEE Trans. Inf. Theory 47(4) (2001), pp. 1473--1481.

\bibitem{EMSS} S. Eddahmani, S. Mesnager, Explicit values of the DDT, the BCT, the FBCT, and the FBDT of the inverse, the gold, and the Bracken-Leander S-boxes, Cryptogr. Commun. (2022), https://doi.org/10.1007/s12095-022-00581-8.


\bibitem{HPSP} S.U. Hasan, M. Pal, P. St\u{a}nic\u{a}, Boomerang uniformity of a class of power maps, Des. Codes Cryptogr. 89 (2021), pp. 2627--2636.
    
\bibitem{HPSP2} S.U. Hasan, M. Pal, P. St\u{a}nic\u{a}, The binary Gold function and its c-boomerang connectivity table, Cryptogr. Commun. (2022), https://doi.org/10.1007/s12095-022-00573-8.

\bibitem{JLLQ} S. Jiang, K. Li, Y. Li, L. Qu, Differential and boomerang spectrums of some power permutations, Cryptogr. Commun. 14 (2022), pp. 371--393.

\bibitem{KMCLJ} K.H. Kim, S. Mesnager, J.H. Choe, D.N. Lee, S. Lee, M.C. Jo, On permutation quadrinomials with boomerang uniformity $4$ and the best-known nonlinearity, Des. Codes Cryptogr. 90 (2022), pp. 1437--1461.

\bibitem{LRFC} L. Lei, W. Ren, C. Fan, The differential spectrum of a class of power functions over finite fields, Adv. Math. Commun. 15(3) (2021), pp. 525--537.

\bibitem{LHXZ} N. Li, Z. Hu, M. Xiong, X. Zeng, A note on ``Cryptographically strong permutations from the butterfly structure", Des. Codes Cryptogr. 90 (2022), pp. 265--276.

\bibitem{LWZT} N. Li, Y. Wu, X. Zeng, X. Tang, On the differential spectrum of a class of power functions over finite fields, arXiv:2012.04316, 2020.

\bibitem{LXZX} N. Li, M. Xiong, X. Zeng, On permutation quadrinomials and 4-uniform BCT, IEEE Trans. Inf. Theory 67(7) (2021), pp. 4845--4855.

\bibitem{LLHQ} K. Li, C. Li, T. Helleseth, L. Qu, Cryptographically strong permutations from the butterfly structure, Des. Codes Cryptogr. 89 (2021), pp. 737--761.

\bibitem{LQSL} K. Li, L. Qu, B. Sun, C. Li, New results about the boomerang uniformity of permutation polynomials, IEEE Trans. Inf. Theory 65(11) (2019), pp. 7542--7553.


\bibitem{Lidl} R. Lidl, H. Niederreiter, Finite Fields, Encyclopedia of Mathematics, vol. 20, Cambridge University Press, Cambridge, 1997.

\bibitem{MXLH} Y. Man, Y. Xia, C. Li, T. Helleseth, On the differential properties of the power mapping $x^{p^m+2}$, 	arXiv:2204.08118, 2022.

\bibitem{MTXM} S. Mesnager, C. Tang, M. Xiong, On the boomerang uniformity of quadratic permutations, Des. Codes Cryptogr. 88(10) (2020), pp. 2233--2246.
    
\bibitem{MMMM} S. Mesnager, B. Mandal, M. Msahli, Survey on recent trends towards generalized differential and boomerang uniformities, Cryptogr. Commun. 14 (2022), pp. 691--735.

\bibitem{NK} K. Nyberg, Differentially uniform mappings for cryptography, in EUROCRYPT 1993, Tor Helleseth, editor, LNCS, vol. 765, pp.~134--144. Springer, Berlin, Heidelberg, May 1994.

\bibitem{PLZX} T. Pang, N. Li, X. Zeng, On the differential spectrum of a differentially 3-uniform power function, IACR Cryptol. ePrint Arch., 2022/610, 2022, https://eprint.iacr.org/2022/610.

\bibitem{TDXM} C. Tang, C. Ding, M. Xiong, Codes, differentially $\delta$-uniform functions, and $t$-designs, IEEE Trans. Inf. Theory 66(6) (2020), pp. 3691--3703.

\bibitem{TZ2019} Z. Tu, X. Zeng, A class of permutation trinomials over finite fields of odd characteristic, Cryptogr. Commun. 11 (2019), pp. 563--583.

\bibitem{TZLH} Z. Tu, X. Zeng, C. Li, T. Helleseth, A class of new permutation trinomials, Finite Fields Appl. 50 (2018), pp. 178--195.

\bibitem{WAG} D. Wagner, The boomerang attack, in FSE 1999, Lars R. Knudsen, editor, LNCS, vol. 1636, pp. 156--170. Springer, Berlin, Heidelberg, March 1999.

\bibitem{XZLH} Y. Xia, X. Zhang, C. Li, T. Helleseth, The differential spectrum of a ternary power mapping, Finite Fields Appl. 64 (2020), 101660.

\bibitem{XYHM} M. Xiong, H. Yan, A note on the differential spectrum of a differentially $4$-uniform power function, Finite Fields Appl. 48 (2017), pp. 117--125.

\bibitem{XYYP} M. Xiong, H. Yan, P. Yuan, On a conjecture of differentially $8$-uniform power functions, Des. Codes Cryptogr. 86(8) (2018), pp. 1601--1621.

\bibitem{YLCH} H. Yan, C. Li, Differential spectra of a class of power permutations with characteristic 5, Des. Codes Cryptogr. 89 (2021), pp. 1181--1191.

\bibitem{YLSF} H. Yan, Z. Li, Z. Song, R. Feng, Two classes of power mappings with boomerang uniformity 2, Adv. Math. Commun. (2022), doi: 10.3934/amc.2022046.
    
\bibitem{YXLHML} H. Yan, Y. Xia, C. Li, T. Helleseth, M. Xiong, J. Luo, The differential spectrum of the power mapping $x^{p^n-3}$, IEEE Trans. Inf. Theory, doi: 10.1109/TIT.2022.3162334.


\bibitem{YZWWHW} H. Yan, Z. Zhou, J. Wen, J. Weng, T. Helleseth, Q. Wang, Differential spectrum of Kasami power permutations over odd characteristic finite fields, IEEE Trans. Inf. Theory 65(10) (2019), pp. 6819--6826.
    
\bibitem{YZZZ} H. Yan, Z. Zhang, Z. Li, Boomerang spectrum of a class of power functions, accepted by IWSDA 2022.

\bibitem{YZZZ1} H. Yan, Z. Zhang, Z. Zhou, A class of power mappings with low boomerang uniformity, accepted by WAIFI 2022.

\bibitem{ZHLZ} Z. Zha, L. Hu, The boomerang uniformity of power permutations $x^{2^{k}-1}$ over $\ftwon$, in Ninth International Workshop on Signal Design and its Applications in Communications (IWSDA), 2019, pp. 1--4.




\end{thebibliography}
\end{document}